\documentclass[fleqn,11p,times]{elsarticle}
\usepackage[utf8]{inputenc}

\usepackage{natbib}
\usepackage{setspace}
\usepackage{graphicx}
\usepackage{graphics}
\usepackage{xcolor}
\definecolor{navy}{rgb}{0.000000,0.000000,0.501961}
\graphicspath{{Images/}}
\usepackage[labelfont=bf]{caption}
\usepackage{subcaption}

\usepackage[ruled,vlined]{algorithm2e}

\usepackage{amsmath}

\usepackage{amsthm}
\theoremstyle{definition}
\newtheorem{definition}{Definition}
\newtheorem{theorem}{Theorem}

\makeatletter
\def\els@aparagraph[#1]#2{\elsparagraph[#1]{#2\@addpunct{.}}}
\def\els@bparagraph#1{\elsparagraph*{#1\@addpunct{.}}}
\SetAlgoCaptionSeparator{.}
\RestyleAlgo{algoruled}

\makeatother

\usepackage{float}



\begin{document}

\begin{frontmatter}

\title{MLN-EIGS: A multilayer network framework for solving Stackelberg escape interdiction games on dynamic transportation networks}

\vspace{0.4cm}
\author[1]{Sukanya Samanta\corref{cor1}}
\ead{susamanta1@gmail.com}
\cortext[cor1]{Corresponding author}

\author[1]{Kei Kimura}

\author[1]{Makoto Yokoo}

\author[2]{Palash Dey}

\vspace{0.4cm}

\address[1]{Department of Informatics, Information Science and Electrical Engineering (ISEE), Kyushu University, Fukuoka, 819-0395, Japan}

\address[2]{Department of Computer Science and Engineering, Indian Institute of Technology Kharagpur, Kharagpur, West Bengal 721302, India}

\vspace{0.4cm}

\begin{abstract}

Interdicting an escaping criminal with limited police resources on large-scale transportation networks is a challenging problem due to the dynamic nature of both attacker movement and defender deployment. This paper proposes \emph{MLN-EIGS}, a multilayer network-based framework for solving dynamic escape interdiction problems formulated as a Stackelberg security game. A time-expanded multilayer network is constructed to explicitly model the temporal evolution of the transportation network and the feasible movements of both the attacker and the defenders. The attacker seeks to maximize the probability of successful escape, while the defenders aim to maximize the probability of interdiction. To efficiently compute the attacker's best response, the probabilistic escape formulation is transformed into an equivalent shortest-path problem through a logarithmic transformation, enabling the use of Dijkstra's algorithm. Since the defender best-response problem is computationally intractable, an approximation defender oracle is developed to generate high-quality defender strategies on the multilayer network. The proposed MLN-EIGS framework is benchmarked against an exact mixed-integer linear programming (MILP)-based Stackelberg formulation on a large real-world transportation network. Computational experiments demonstrate that MLN-EIGS consistently achieves defender utilities that closely match those of the exact MILP approach while substantially reducing computational time. These results demonstrate that the proposed MLN-EIGS framework provides an effective, computationally efficient, and scalable alternative to exact MILP-based Stackelberg optimization for large-scale dynamic escape interdiction problems.

\end{abstract}

\begin{keyword}
Stackelberg security games, escape interdiction, multilayer time-expanded network, exact attacker oracle, approximation defender oracle, transportation networks
\end{keyword}
\end{frontmatter}

\section{Introduction}
\label{S:1}

Escape interdiction is a fundamental problem in public safety and transportation security, where law enforcement agencies attempt to intercept an escaping criminal using limited police resources deployed over a large transportation network. The dynamic nature of the problem, together with the temporal evolution of attacker and defender movements, makes strategy generation computationally challenging. Existing approaches are largely based on static network models or exact mixed-integer linear programming (MILP) formulations, which become computationally expensive for large-scale transportation networks.

In this paper, we formulate the dynamic escape interdiction problem as a Stackelberg security game in which the defender acts as the leader and the attacker acts as the follower. The defenders are assumed to know only the crime location and must determine patrol strategies that maximize the probability of intercepting the attacker before reaching one of several predefined exit nodes. After observing the defender's committed mixed strategy, the attacker computes an optimal escape route that minimizes the probability of interdiction.

To efficiently model the temporal evolution of both players, we introduce a multilayer time-expanded transportation network in which each layer represents the transportation network at a particular time instant. The proposed framework, referred to as \textbf{MLN-EIGS} (Multi-Layer Network Escape Interdiction Game Solver), combines an exact attacker oracle with an approximation defender oracle for the defender within a double-oracle framework. The attacker best-response problem is reformulated through a logarithmic transformation into an equivalent shortest-path problem and solved exactly using Dijkstra's algorithm on the multilayer network. Since the defender best-response problem is computationally intractable (see Section~\ref{S:4.3}), an approximation defender oracle is developed to generate high-quality defender strategies while maintaining computational scalability.

Compared with classical zero-sum formulations, the Stackelberg framework more accurately models practical escape interdiction scenarios because the defenders commit to their mixed strategy before the attacker selects an escape route. This sequential decision-making process enables the attacker to compute a best response to the defender's strategy and allows the equilibrium to capture the inherent leader--follower interaction between law enforcement agencies and criminals.

The main contributions of this paper are summarized as follows.

\begin{itemize}
    \item We propose \textbf{MLN-EIGS}, a novel multilayer network-based Stackelberg escape interdiction framework that integrates temporal network expansion, probabilistic interdiction modeling, and Stackelberg game theory into a unified optimization framework for dynamic transportation networks.

    \item We develop an \textbf{exact attacker oracle} by reformulating the multiplicative escape-probability objective into an equivalent shortest-path problem through a logarithmic transformation, enabling efficient computation using Dijkstra's algorithm on the multilayer network.

    \item We develop a \textbf{polynomial-time approximation defender oracle} for the defender best-response problem within the double-oracle framework, allowing scalable equilibrium computation on large transportation networks.

    \item We provide a \textbf{theoretical analysis} of the proposed framework, including the computational complexity of the attacker oracle and the computational hardness of the defender best-response problem.

    \item We conduct extensive \textbf{computational experiments} on a large real-world transportation network and demonstrate that the proposed MLN-EIGS framework achieves defender utilities comparable to those of the benchmark MILP-EIGS formulation while significantly reducing computational time.
\end{itemize}

The primary novelty of this work lies in integrating dynamic probabilistic interdiction modeling with a Stackelberg security game framework on a multilayer time-expanded transportation network. Unlike existing approaches that primarily consider static security games or exact MILP-based interdiction models, the proposed framework simultaneously captures (i) temporal feasibility of attacker and defender movements, (ii) multiplicative probabilistic interception along escape paths, and (iii) computationally scalable equilibrium computation for large-scale transportation networks. By combining a multilayer network representation with a logarithmic reformulation of the escape probability and a double-oracle solution framework, the proposed approach enables exact computation of the attacker best response through shortest-path optimization while efficiently approximating the defender best response. Consequently, the proposed framework achieves substantially improved computational scalability compared with conventional MILP-based formulations while maintaining competitive solution quality.

The remainder of this paper is organized as follows. Section~\ref{S:2} reviews the related literature. Section~\ref{S:3} presents the problem formulation and mathematical model. Section~\ref{S:4} describes the proposed MLN-EIGS framework, including the attacker and defender strategy-generation algorithms. The benchmark MILP-EIGS formulation is presented in Section~\ref{S:5}. Computational experiments and performance comparisons are reported in Section~\ref{S:6}. Finally, Section~\ref{S:7} concludes the paper and outlines directions for future research.

\section{Related work}
\label{S:2}

Security games have received considerable attention because of their wide range of applications in public safety, transportation security, wildlife protection, and critical infrastructure defense (\citet{hunt2024review}, \citet{samanta2022literature}). Among them, Stackelberg security games have become one of the most widely adopted frameworks for modeling strategic interactions between defenders and attackers, where the defender commits to a strategy before the attacker responds. Several studies have investigated escape interdiction and patrolling problems on transportation networks using Stackelberg formulations. For example, \citet{zychowski2023coevolution} consider search games on directed graphs with multiple defenders and propose a genetic algorithm for large synthetic instances. \citet{basilico2009leader} formulate leader--follower security games as mathematical programming models to compute equilibrium strategies. Likewise, \citet{letchford2013solving} develop polynomial-time algorithms for graph-based security games, while \citet{iwashita2016simplifying} improve computational scalability by reducing graph size through edge elimination. Other representative studies include Monte Carlo Tree Search for patrol planning (\citet{karwowski2019monte}) and repeated Stackelberg security games incorporating human behavioral models (\citet{wang2019repeated}).

Several studies have focused on exact optimization techniques for computing Stackelberg equilibria. For example, \citet{tsai2010urban} formulate urban security allocation as a linear programming problem to compute optimal mixed strategies under limited defender resources. Similar optimization-based frameworks have been developed for maritime security (\citet{shieh2012protect}) and other security applications (\citet{cermak2016using}, \citet{lou2017multidefender}, \citet{sinha2018stackelberg}, \citet{vcerny2018incremental}, \citet{zhang2021bayesian}). The computational foundations of optimal commitment strategies in Stackelberg games were established by \citet{conitzer2006computing}. Subsequently, \citet{paruchuri2008playing} developed a MILP-based approach for Bayesian Stackelberg games, while \citet{bosansky2015sequence} proposed scalable sequence-form algorithms for computing Strong Stackelberg Equilibria. Although these exact optimization approaches compute high-quality equilibrium solutions, their computational requirements often limit their applicability to large-scale transportation networks.

Escape interdiction has also been extensively studied under zero-sum game formulations. For example, \citet{zhang2017optimal} propose an exact MILP formulation for escape interdiction and demonstrate its effectiveness on grid networks. To improve computational scalability, \citet{samanta2022vns} develop a metaheuristic-based solution approach, whereas \citet{samanta2021vehicle} propose a simulation-based framework for escape interdiction on large transportation networks. While these approaches improve computational performance, they primarily consider static or deterministic formulations and do not explicitly model temporal network evolution within a Stackelberg framework.

Time-expanded and multilayer network representations have proven effective for solving dynamic optimization problems on transportation networks. For instance, \citet{saito2009discovering} employ layered graphs to identify important nodes in dynamic networks by explicitly modeling temporal evolution through multiple network layers. Such representations naturally capture time-dependent movement constraints and have become an effective modeling paradigm for dynamic transportation problems.

Despite these advances, existing studies either focus on static Stackelberg security games, deterministic escape interdiction models, or computationally expensive MILP formulations. To the best of our knowledge, no existing work simultaneously integrates (i) temporal feasibility of attacker and defender movements, (ii) probabilistic escape and interdiction modeling, and (iii) computationally scalable Stackelberg equilibrium computation on large-scale transportation networks. The proposed MLN-EIGS framework addresses this gap by combining a multilayer time-expanded network representation with an exact attacker oracle, an approximation defender oracle for the defender, and a double-oracle Stackelberg solution framework.

\section{Problem description and modeling}
\label{S:3}

We consider a two-player Stackelberg security game played between a team of defenders and a single attacker. The defender acts as the leader and commits to a mixed strategy before the attacker, who acts as the follower, selects an optimal escape route in response. Although multiple police units participate in the interdiction process, they cooperate as a single decision maker and therefore constitute one player in the Stackelberg game.

Let $\overline{D}=\{d_r \mid r\in R\}$ denote the set of defenders, where
$R=\{1,\ldots,m\}$ and $m$ is the total number of defenders. The attacker is denoted by $\overline{A}$. The defender has a finite set of pure strategies $\Acute{S}$, while the attacker has a finite set of pure strategies $\Acute{A}$. Let $x$ and $y$ denote the corresponding mixed strategies, i.e., probability distributions over $\Acute{S}$ and $\Acute{A}$, respectively.

The transportation network is represented by a directed graph
$G=(V,E)$, where $V$ denotes the set of intersections and $E$ denotes the set of directed road segments. A subset of nodes represents predefined exit points through which the attacker may leave the transportation network. Let $v_{\infty}$ denote an arbitrary exit node. The game begins at time $0$ and terminates at time $t_{\max}>0$.

A pure attacker strategy is a time-ordered sequence of node--time pairs

\[
A=
\langle
a_1=(v_0^a,0),
\ldots,
a_j=(v_j,t_j^a),
\ldots,
a_k=(v_{\infty},t_k^a)
\rangle,
\]

where $t_k^a\le t_{\max}$. Each state
$a_j=(v_j,t_j^a)$
indicates that the attacker occupies node $v_j$ at time $t_j^a$.

A pure defender strategy $S$ consists of one patrol schedule for each defender,

\[
S=\{S^r:r\in R\},
\]

where the patrol schedule of defender $d_r$ is

\[
S^r=
\langle
s_1^r,\ldots,s_i^r,\ldots,s_k^r
\rangle,
\]

with

\[
s_i^r=(v_i^r,t_{i}^{r,\mathrm{in}},t_{i}^{r,\mathrm{out}}).
\]

Each state specifies that defender $d_r$ occupies node $v_i^r$ during the time interval

\[
[t_i^{r,\mathrm{in}},\,t_i^{r,\mathrm{out}}].
\]

The defender's mixed strategy is denoted by

\[
x=\{x_S:S\in\Acute{S}\},
\]

where $x_S$ is the probability assigned to pure strategy $S$.

Defender $d_r$ intercepts the attacker whenever both occupy the same node at the same time, i.e.,

\[
v_i^r=v_j
\]

and

\[
t_i^{r,\mathrm{in}}
\le
t_j^a
\le
t_i^{r,\mathrm{out}}.
\]

Since the defender receives a payoff of one if the attacker is intercepted and zero otherwise, the defender's expected utility under mixed strategy $x$ and attacker pure strategy $A$ equals the probability of successful interdiction:

\[
U_d(x,A)
=
\Pr(\text{attacker is intercepted}\mid x,A).
\]

Similarly, the attacker's utility equals the probability of successfully escaping the transportation network:

\[
U_a(x,A)
=
\Pr(\text{attacker escapes}\mid x,A).
\]

Since interception and escape are complementary events,

\[
U_a(x,A)
=
1-U_d(x,A).
\]

Consequently, maximizing the defender's utility is equivalent to maximizing the probability of interdiction, whereas maximizing the attacker's utility is equivalent to maximizing the probability of successful escape.

When the attacker adopts a mixed strategy $y$, the defender's expected utility is denoted by $U_d(x,y)$ and is obtained as the expectation of $U_d(x,A)$ over all attacker pure strategies.

Given the attacker's pure strategy set, the defender computes an optimal mixed strategy by solving the following restricted master linear program:

\begin{equation}
\max \quad U
\end{equation}

\begin{equation}
\text{s.t.}\quad
U
\le
U_d(x,A),
\qquad
\forall A\in\Acute{A},
\end{equation}

\begin{equation}
\sum_{S\in\Acute{S}}x_S=1,
\qquad
x_S\ge0,
\qquad
\forall S\in\Acute{S}.
\end{equation}

Given a defender mixed strategy $x$, the attacker computes a best response from the set

\begin{equation}
BR(x)
=
\arg\min_{y\in\Acute{A}}
U_d(x,y).
\end{equation}

The defender then seeks a mixed strategy that maximizes its expected utility against the attacker's best response:

\begin{equation}
\max_{x\in\Acute{S}}
U_d(x,y)
\qquad
\text{s.t.}
\qquad
y\in BR(x).
\end{equation}

The solution concept adopted in this paper is the \emph{Strong Stackelberg Equilibrium (SSE)}, which assumes that the follower breaks ties in favor of the leader and is the standard solution concept for Stackelberg security games
\citep{conitzer2006computing,kroer2022lecture}. The resulting optimization problem is

\begin{equation}
\max_{x\in\Acute{S},\,y\in BR(x)}
U_d(x,y).
\end{equation}

\section{Proposed MLN-EIGS framework}
\label{S:4}

This section presents the proposed MLN-EIGS framework for solving the escape interdiction problem formulated as a Stackelberg security game. The defender acts as the leader and first commits to a mixed strategy, after which the attacker computes an optimal escape route as its best response. To efficiently model the temporal evolution of attacker and defender movements, we construct a multilayer time-expanded transportation network (MLN), where each layer represents the transportation network at a discrete time instant. Successive layers are connected according to the travel time associated with each road segment, thereby preserving the temporal feasibility of all movements.

The proposed framework adopts a double-oracle solution methodology that combines an exact attacker oracle with an efficient approximation defender oracle. The attacker best-response problem is solved exactly, whereas the defender best-response problem is approximated to maintain computational scalability on large transportation networks. The defender commits to a mixed strategy, represented as a probability distribution over a finite set of pure strategies, whereas the attacker computes an exact best response over its pure strategy space. To compute the attacker’s optimal strategy, the interdiction probability $P(n,t)$, induced solely by the defender’s committed mixed strategy, is assigned to each node--time pair in the multilayer network. Since the attacker is intercepted if interdiction occurs at any node along its path, the escape probability is transformed into an equivalent additive objective through a logarithmic transformation. The resulting non-negative transformed node weights enable Dijkstra's shortest-path algorithm to efficiently compute the attacker’s optimal escape path.

The multilayer network representation is essential for modeling the temporal nature of the escape interdiction problem. Because travel times are associated with network edges, the feasibility of both attacker and defender movements depends on both spatial connectivity and arrival times. A conventional single-layer graph cannot distinguish between visits to the same node at different times, resulting in an incorrect evaluation of interception opportunities. By explicitly expanding the transportation network over time, the multilayer representation preserves temporal feasibility while transforming the dynamic path-planning problem into an equivalent static shortest-path problem. Furthermore, it provides a principled mechanism for aggregating interdiction probabilities under mixed defender strategies, thereby ensuring both modeling accuracy and computational efficiency.

Algorithm~\ref{alg:a1} summarizes the overall MLN-EIGS framework. At each iteration, the restricted master problem (\textsc{RestrictedStackelbergLP}) computes the defender's optimal mixed strategy $x^{\star}$ by solving the linear program defined in Eqs.~(1)--(3) over the current restricted strategy sets. The attacker then computes an exact best response using the proposed attacker oracle (\textsc{ExactAO}), while the defender generates an improving strategy using the approximation defender oracle (\textsc{ApproxDO}). Newly generated attacker or defender strategies are added to the corresponding restricted strategy sets, and the restricted Stackelberg game is solved again.

For the defender oracle, a time-expanded multilayer network is constructed in the same manner as for the attacker. Node weights are determined from the current attacker strategy set, where each attacker strategy is assigned equal probability. The defender best-response problem is then approximated by identifying a patrol path that covers as many distinct attacker strategies as possible. To facilitate this process, each attacker strategy is assigned a unique color, and the approximation defender oracle seeks a feasible defender path that visits the maximum number of differently colored nodes, thereby maximizing the expected probability of interdiction.

The iterative strategy-generation process terminates when neither the attacker nor the defender can generate an improving strategy outside the current restricted strategy sets. Since the attacker best-response problem is solved exactly while the defender best-response problem is approximated, the proposed framework computes an \emph{approximate Strong Stackelberg Equilibrium}. The resulting equilibrium strategy pair $(x^{\star},A^{\star})$ and the corresponding defender utility $U_d^{\star}$ constitute the final solution returned by the MLN-EIGS framework.

\begin{algorithm}[H]
\SetAlgoLined
 \textbf{Input:} Initial restricted defender and attacker strategy sets $S'$ and $A'$.\;
 \textbf{Output:} Approximate Stackelberg equilibrium $(x^{\star},A^{\star})$ and defender utility $U_d^{\star}$.\;
 \Repeat{no improving strategy exists}{
 
        $(x^{\star}) \leftarrow RestrictedStackelbergLP (S^{'}, A^{'})$\;

        Compute defender mixed strategy $x^{\star}$.\;
            
		$ BR: \; A^{\star} \leftarrow { ExactAO(x^{\star}) }$\;

		\If{$A^{\star} \notin A^{'} $}{
   			$A^{'} \leftarrow A^{'} \cup \{A^{\star}\}  $\;
            \textbf{continue}\;
   			}

 		$ S^{\star} \leftarrow {ApproxDO(A^{'}) }$\;

		\If{$S^{\star} \notin S^{'} $}{
   			$S^{'} \leftarrow S^{'} \cup \{S^{\star}\}  $\;
            \textbf{continue}\;
   			} 
    }
$U_d^{\star} \leftarrow U_d(x^{\star}, A^{\star})$\;
\tcp{Final defender utility at equilibrium}

\Return $U_d^{\star}, (x^{\star}, A^{\star})$.

 \caption{Stackelberg game solved using the MLN-EIGS algorithm.}
 \label{alg:a1}
\end{algorithm}

Since the attacker best-response problem is solved exactly whereas the defender best-response problem is generated using an approximation defender oracle, the proposed MLN-EIGS framework computes an \emph{approximate Strong Stackelberg Equilibrium}. Consequently, the resulting strategy pair $(x^{\star},A^{\star})$ represents an approximate equilibrium of the original game, and the corresponding defender utility $U_d^{\star}$ is reported as the approximate Stackelberg equilibrium utility.

\subsection{Efficient attacker strategy generation using an exact approach on the time-expanded network}
\label{S:4.1}

The attacker best-response problem is solved exactly on the proposed multilayer time-expanded transportation network. Algorithm~\ref{alg:a3} summarizes the proposed attacker oracle. A copy of the original transportation network is created for each discrete time step, forming a multilayer network in which temporal feasibility is explicitly represented. Successive layers are connected according to the travel times (edge lengths) of the corresponding road segments in the original network. Initially, the interdiction probability and transformed weight of every node--time pair are initialized to zero. Based on the defender's committed mixed strategy, interdiction probabilities are computed for each node--time pair and subsequently transformed into node weights using the proposed logarithmic transformation. These transformed weights are assigned to the corresponding incoming edges, after which Dijkstra's shortest-path algorithm is applied to compute the attacker's optimal escape path.

In computing the attacker’s best response, it is essential to model the probabilistic structure of sequential interdiction correctly. The attacker is intercepted if interdiction occurs at \emph{any} node along the selected path, which corresponds to the union of interception events over all visited nodes. The interdiction probability at each node--time pair $(n,t)$ equals the total probability mass of defender strategies that occupy that node at time $t$. 

Let $P(n,t)$ denote the probability that the attacker is interdicted at node $n$ and time $t$ under the defender's mixed strategy $x$. Since the defender commits to a mixed strategy over the set of pure defender strategies $\Acute{S}$, the interdiction probability at each node--time pair is computed as

\begin{equation}
P(n,t)=\sum_{S\in\Acute{S}}x_S\,I_S(n,t),
\label{eq:interdiction_probability}
\end{equation}
where $x_S$ denotes the probability assigned to defender strategy $S$, satisfying
$\sum_{S\in\Acute{S}}x_S=1$, and

\[
I_S(n,t)=
\begin{cases}
1,&\text{if defender strategy }S\text{ occupies node }n\text{ at time }t,\\
0,&\text{otherwise.}
\end{cases}
\]

The attacker escapes only if interception does not occur at any node--time pair along the selected path. To obtain a computationally tractable optimization model, we assume that interdiction events at distinct node--time pairs are conditionally independent given the defender's mixed strategy. Under this assumption, the attacker's overall escape probability can be expressed as the product of the survival probabilities associated with the visited node--time pairs. Accordingly, the attacker's probability of successfully escaping along path $A$ is given by

\begin{equation}
U_a(x,A)=
\prod_{(n,t)\in A}\left(1-P(n,t)\right),
\label{eq:attacker_utility}
\end{equation}

where $U_a(x,A)$ denotes the attacker's expected utility. Since successful interdiction and successful escape are complementary events, the defender's expected utility is

\begin{equation}
U_d(x,A)=
1-
\prod_{(n,t)\in A}\left(1-P(n,t)\right).
\label{eq:defender_utility}
\end{equation}

The logarithmic transformation converts the multiplicative escape probability into an additive path cost. Specifically, the transformed weight associated with each node--time pair is defined as

\begin{equation}
w(n,t)=-\log\left(1-P(n,t)\right).
\label{eq:log_transform}
\end{equation}

Since $0 \le P(n,t) < 1$, it follows that $w(n,t)\ge 0$ for every node--time pair. Consequently, the transformed multilayer network contains only non-negative edge weights, allowing Dijkstra's shortest-path algorithm to be applied directly. Therefore, computing the attacker's best response is equivalent to finding the minimum-cost path from the crime node to any exit node in the weighted multilayer network. The overall workflow of the proposed attacker oracle is summarized in Algorithm~\ref{alg:a3}.

\begin{algorithm}[H]
\SetAlgoLined
\textbf{Input:} Crime node (START), exit nodes (GOAL), original graph $G_0$, defender mixed strategy $x$\;

\textbf{Output:} Attacker best-response strategy with minimum transformed path cost\;

\BlankLine
\textbf{Construct multilayer network:}\\
\For{$i=0$ \KwTo $t_{\max}$}{
Generate one copy of the original graph, denoted by $G_i$\;
}

\BlankLine
\textbf{Construct temporal edges:}\\
\For{each edge $(S,T)$ in the original graph}{
Let $L$ denote the travel time of edge $(S,T)$\;
\For{$j=0$ \KwTo $t_{\max}-L$}{
Create an edge from node $S$ in layer $G_j$ to node $T$ in layer $G_{j+L}$\;
}
}

\BlankLine
Initialize $P(n,t)=0$ for every node--time pair $(n,t)$\;

Initialize $w(n,t)=0$ for every node--time pair $(n,t)$\;

\BlankLine
\textbf{Compute node--time interdiction probabilities:}\\
\For{each defender strategy $S\in\Acute{S}$}{
\For{each node visit $(n,t_{\mathrm{in}},t_{\mathrm{out}})$ in strategy $S$}{
\For{$t=t_{\mathrm{in}}$ \KwTo $t_{\mathrm{out}}$}{
\[
P(n,t)\leftarrow P(n,t)+x_S;
\]
}
}
}
\caption{Optimal attacker strategy design using the time-expanded multilayer network.}
\label{alg:a3}
\end{algorithm}

\begin{algorithm}[H]
\SetAlgoLined
\textbf{Compute transformed node weights:}\\
\For{each node--time pair $(n,t)$}{
\[
w(n,t)\leftarrow -\log\!\left(1-P(n,t)\right);
\]
}

\BlankLine
\textbf{Assign transformed node weights to incoming edges:}\\
\For{each node--time pair $(n,t)$}{
Assign weight $w(n,t)$ to every incoming edge of node $(n,t)$\;
}

\BlankLine
Apply Dijkstra's shortest-path algorithm on the weighted multilayer network\;

\Return Attacker path with minimum transformed path cost\;
\end{algorithm}

Since the interception probability is associated with arriving at a node rather than traversing an edge, the transformed node weight is assigned to all incoming edges. Consequently, every node contributes exactly once to the total path cost.

The constructed multilayer network contains $O(|V|t_{\max})$ vertices and $O(|E|t_{\max})$ edges, where $|V|$ and $|E|$ denote the number of vertices and edges in the original transportation network, respectively. Consequently, the attacker best-response problem can be solved in

\[
O\!\left((|E|t_{\max})\log(|V|t_{\max})\right)
\]

time using Dijkstra's shortest-path algorithm with a binary heap. Since the temporal expansion factor is linear in $t_{\max}$, the attacker oracle scales efficiently with the time horizon while avoiding the repeated solution of mixed-integer optimization problems. This represents a substantial reduction in computational complexity compared with the repeated solution of mixed-integer linear programs required by the benchmark MILP-based attacker oracle.

Once the transformed node weights have been assigned, the attacker best-response problem can be solved by applying Dijkstra's shortest-path algorithm to the weighted multilayer network. Specifically,

\[
\sum_{(n,t)\in A}
-\log\left(1-P(n,t)\right)
=
-\log\left(
\prod_{(n,t)\in A}
\left(1-P(n,t)\right)
\right),
\]

minimizing the total transformed path cost is equivalent to maximizing the attacker's escape probability,

\[
\prod_{(n,t)\in A}
\left(1-P(n,t)\right),
\]

or, equivalently, minimizing the defender's interception probability,

\[
1-
\prod_{(n,t)\in A}
\left(1-P(n,t)\right).
\]

Therefore, the shortest path returned by Dijkstra's algorithm corresponds exactly to the attacker's optimal best-response strategy under the proposed probabilistic interdiction model.


\subsection{Efficient defender strategy generation using an approximation defender oracle on the time-expanded network}
\label{S:4.2}

Unlike the attacker best-response problem, which admits an exact shortest-path formulation on the multilayer network, the defender best-response problem is computationally more challenging because the defender must simultaneously maximize the expected interdiction probability over multiple attacker strategies while satisfying temporal movement constraints. As proved in Theorem~\ref{thm:color_covering_npcomplete}, the corresponding decision problem is NP-complete. Consequently, an exact polynomial-time algorithm for computing the defender's optimal strategy is unlikely to exist unless $\mathrm{P}=\mathrm{NP}$. We therefore develop a computationally efficient approximation defender oracle that generates effective defender strategies and enables scalable computation of approximate Stackelberg equilibria on large transportation networks.

The proposed defender oracle is based on a greedy \emph{color-covering} heuristic. Each attacker pure strategy in the current restricted strategy set is assigned a unique color, and every node--time pair visited by that attacker strategy inherits the corresponding color. A node may therefore contain multiple colors if it belongs to several attacker strategies. The defender seeks to construct a temporally feasible patrol path whose total travel time does not exceed the prescribed response-time limit while visiting the maximum number of distinct colors. Since each color represents one attacker strategy, covering additional colors increases the number of attacker strategies that can potentially be intercepted. Therefore, maximizing the number of distinct colors visited provides a natural surrogate objective for maximizing the defender's expected utility.

Algorithm~\ref{alg:defender} summarizes the proposed approximation procedure. Starting from the defender's initial location, the algorithm repeatedly computes a shortest path on the time-expanded multilayer network using Dijkstra's algorithm. At each iteration, the nearest reachable node containing at least one previously uncovered color is selected, the defender's patrol path is extended accordingly, and the corresponding color(s) are marked as covered. This greedy process continues until either all colors have been covered or the available travel-time budget has been exhausted.

The initial patrol path may not cover every attacker strategy. Therefore, additional patrol paths are generated by repeatedly applying the same procedure while prioritizing the remaining uncovered colors. The resulting collection of patrol paths forms the restricted defender strategy set, from which the defender subsequently computes an optimal mixed strategy by solving the restricted master problem.

\begin{algorithm}[H]
\SetAlgoLined

\textbf{Input:} Current attacker strategy set $A'$, defender initial node $v_0$, response-time limit $t_{\max}$\;

\textbf{Output:} Approximate defender strategy $S^{\star}$\;

Assign a unique color to every attacker strategy in $A'$\;

Assign the corresponding color(s) to every node--time pair visited by each attacker strategy\;

Initialize the covered-color set $C\leftarrow\emptyset$\;

Initialize defender path $S^{\star}\leftarrow\{v_0\}$\;

Initialize remaining travel budget $T\leftarrow t_{\max}$\;

\While{$T>0$ \textbf{and} not all colors are covered}{

Compute the shortest paths from the current defender location using Dijkstra's algorithm\;

Select the nearest reachable node containing at least one uncovered color\;

\eIf{the node is reachable within the remaining travel budget}{

Append the shortest path to $S^{\star}$\;

Update the remaining travel budget $T$\;

Mark all colors of the selected node as covered\;

Update the current defender location\;

}{
\textbf{break}\;
}

}

\While{uncovered colors remain}{

Generate another patrol path using the same greedy procedure considering only the remaining uncovered colors\;

Add the generated patrol path to the defender strategy set\;

Update the covered-color set\;

}

Construct the defender mixed strategy by randomizing over the generated patrol paths\;

\Return $S^{\star}$\;

\caption{Approximation defender oracle (ApproxDO) based on greedy color covering.}
\label{alg:defender}

\end{algorithm}

The proposed defender oracle relies exclusively on repeated shortest-path computations on the time-expanded multilayer network, thereby avoiding the repeated solution of computationally expensive mixed-integer optimization problems required by the benchmark MILP-based defender oracle. Let $k$ denote the number of attacker strategies (equivalently, the number of distinct colors) currently contained in the restricted strategy set. Since each iteration performs one shortest-path computation and covers at least one previously uncovered color, at most $k$ shortest-path computations are required. The time-expanded multilayer network contains $O(|V|t_{\max})$ vertices and $O(|E|t_{\max})$ edges, where $|V|$ and $|E|$ denote the numbers of vertices and edges in the original transportation network, respectively. Consequently, using Dijkstra's algorithm with a binary heap, the overall computational complexity of the defender oracle is

\[
O\!\left(k\,|E|\,t_{\max}\log\!\left(|V|\,t_{\max}\right)\right).
\]

Although the proposed defender oracle does not guarantee a globally optimal defender strategy or establish a theoretical approximation ratio, it runs in polynomial time and efficiently generates high-quality patrol strategies. Its effectiveness is demonstrated empirically through the computational experiments presented in Section~\ref{S:6}, where it consistently produces defender equilibrium utilities close to those obtained by the benchmark MILP-EIGS formulation while requiring substantially lower computational time.


\subsection{Computational hardness of the defender best-response problem}
\label{S:4.3}

The computational hardness of the defender best-response problem is established through the following decision problem.

\begin{definition}[Color-Covering Problem]
Given a directed graph, an initial vertex, a set of colored vertices, and a path-length bound $n$, determine whether there exists a path originating from the initial vertex whose length does not exceed $n$ and that visits every color at least once.
\end{definition}

\begin{theorem}
\label{thm:color_covering_npcomplete}
The Color-Covering Problem is NP-complete.
\end{theorem}

\begin{proof}[Proof Sketch]
The proof is obtained by a polynomial-time reduction from the classical 3-SAT problem.

Consider an arbitrary 3-SAT instance consisting of $n$ Boolean variables $(x_1,\ldots,x_n)$ and $m$ clauses. We construct a layered directed graph with $n+1$ levels. Level $0$ contains a single source vertex. For each variable $x_i$ $(i=1,\ldots,n)$, level $i$ contains two vertices corresponding to the assignments $x_i=\texttt{true}$ and $x_i=\texttt{false}$. Every vertex at level $i$ is connected by directed edges to both vertices at level $i+1$, thereby representing every possible truth assignment as a root-to-leaf path.

Each clause is assigned a unique color. If a clause contains the literal $x_i$, then the vertex corresponding to the assignment $x_i=\texttt{true}$ is assigned the color of that clause. Similarly, if a clause contains the literal $\neg x_i$, then the vertex corresponding to the assignment $x_i=\texttt{false}$ receives the corresponding color.

Each root-to-leaf path uniquely represents a truth assignment of the Boolean variables. Moreover, a clause color is visited if and only if the corresponding clause is satisfied by that assignment. Consequently, the constructed graph contains a path covering every color if and only if the original 3-SAT instance is satisfiable.

Since the reduction can be performed in polynomial time and the Color-Covering Problem belongs to NP, the Color-Covering Problem is NP-complete.
\end{proof}

\section{MILP-EIGS benchmarking algorithm}
\label{S:5}

To evaluate the effectiveness of the proposed MLN-EIGS framework, we employ an exact Stackelberg game formulation based on mixed-integer linear programming (MILP) as the benchmark. The benchmark, referred to as \textit{MILP-EIGS}, follows the same leader--follower framework as MLN-EIGS but computes both the attacker and defender best responses exactly using the MILP formulations proposed by \citet{zhang2017optimal}. Specifically, the attacker oracle (\textit{bestAo}) computes the exact attacker best response, while the defender oracle (\textit{bestDo}) generates the exact defender best response.

Algorithm~\ref{alg:a2} summarizes the MILP-EIGS framework. Similar to MLN-EIGS, the algorithm iteratively solves the restricted Stackelberg game using the current attacker and defender strategy sets. At each iteration, the restricted master problem computes the defender's optimal mixed strategy. The attacker then generates an exact best response using the MILP-based attacker oracle (\textit{bestAo}). If a new attacker strategy is identified, it is added to the restricted strategy set and the restricted game is solved again. Otherwise, the defender computes an exact best response using the MILP-based defender oracle (\textit{bestDo}). The algorithm terminates when neither player can generate a new improving strategy, yielding a Stackelberg equilibrium of the restricted game.

Since both the attacker and defender best-response problems are solved exactly, the resulting strategy pair $(x^{\star},A^{\star})$ satisfies the equilibrium conditions of the Stackelberg game. However, both MILP oracles repeatedly solve mixed-integer optimization problems during the iterative strategy-generation process. As reported by \citet{zhang2017optimal}, the vehicle escape interdiction problem is NP-hard, and the repeated solution of MILPs results in substantial computational overhead for large transportation networks.

The MILP-EIGS framework is therefore used as the benchmark for evaluating the proposed MLN-EIGS algorithm. In the experimental study, the two approaches are compared in terms of defender equilibrium utility, computational time, and optimality gap.

\begin{algorithm}[H]
\SetAlgoLined

\textbf{Input:} Initial restricted defender and attacker strategy sets $S'$ and $A'$\;

\textbf{Output:} Stackelberg equilibrium $(x^{\star},A^{\star})$ and defender utility $U_d^{\star}$\;

\Repeat{no improving strategy exists}{

$x^{\star}\leftarrow
RestrictedStackelbergLP(S',A')$\;

$A^{\star}\leftarrow bestAo(x^{\star})$\;

\If{$A^{\star}\notin A'$}{
$A'\leftarrow A'\cup\{A^{\star}\}$\;
\textbf{continue}\;
}

$S^{\star}\leftarrow bestDo(A')$\;

\If{$S^{\star}\notin S'$}{
$S'\leftarrow S'\cup\{S^{\star}\}$\;
\textbf{continue}\;
}

}

$U_d^{\star}\leftarrow U_d(x^{\star},A^{\star})$\;

\Return $(x^{\star},A^{\star}),U_d^{\star}$\;

\caption{MILP-EIGS benchmarking algorithm.}
\label{alg:a2}

\end{algorithm}

\section{Results and discussion}
\label{S:6}

In this section, we evaluate the performance of the proposed MLN-EIGS framework. All experiments were implemented in Python~3.6 and executed on an Ubuntu~18.04.6 LTS workstation equipped with an Intel Core i5-10300H processor, 15~GB RAM, and IBM CPLEX~12.8 as the optimization solver.

To illustrate the proposed attacker oracle, we first consider the four-node transportation network shown in Fig.~\ref{image-1}. In this example, nodes~2 and~3 represent police stations, node~1 is the crime node, node~4 is the exit node, and the maximum response time is $t_{\max}=5$. Here, $0\_4$ denotes node~4 at time $t=0$ in the multilayer network. Two defender strategies with mixed probabilities of $1/3$ and $2/3$ are assumed, from which the node--time interdiction probabilities are computed. These probabilities are transformed into node weights according to the proposed logarithmic formulation, and Dijkstra's algorithm is applied to the time-expanded multilayer network (Fig.~\ref{image-4}) to compute the attacker's exact best response. The resulting optimal escape path, shown by the red trajectory in Fig.~\ref{image-4}, is
\[
0\_1 \rightarrow 3\_3 \rightarrow 5\_4.
\]
The corresponding attacker strategy and computational time are reported in Table~\ref{tab:table1}. This illustrative example demonstrates that the proposed attacker oracle efficiently computes the optimal escape strategy on the multilayer network while preserving the probabilistic interdiction model.

\begin{figure}[H]
\centering
\includegraphics[scale=0.35]{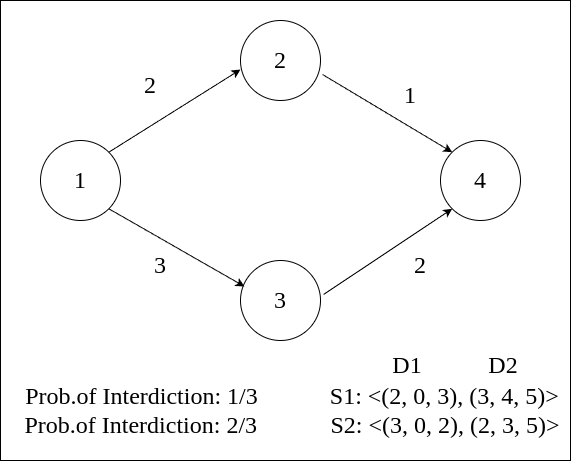}
\caption{Sample network for designing the optimal attacker strategy.}
\label{image-1}
\end{figure} 

\begin{figure}[H]
\centering
\includegraphics[scale=0.36]{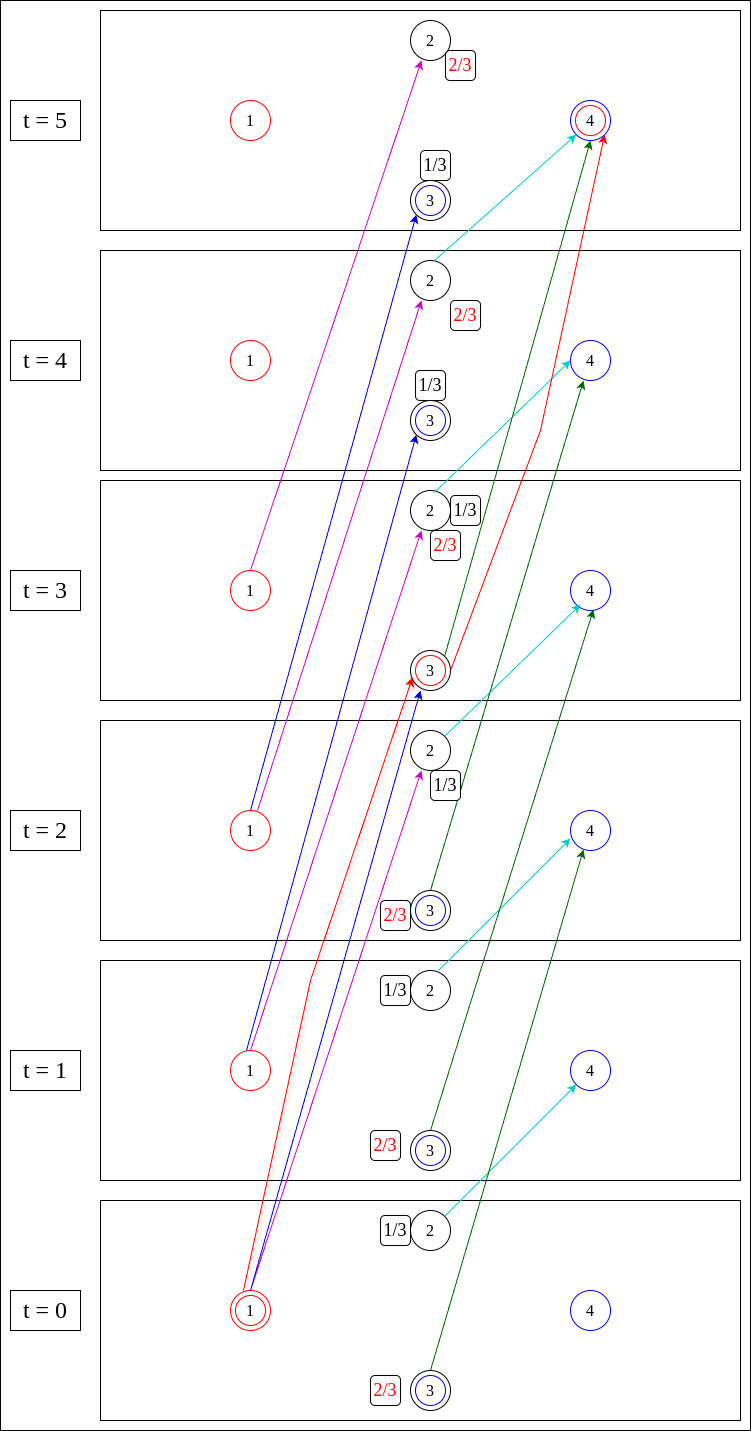}
\caption{Design of a multi-layer network for attacker considering a sample network (Fig. \ref{image-1}).}
\label{image-4}
\end{figure}

\begin{table}[H]
  \begin{center}
   \caption{\\Optimal attacker strategy design using Dijkstra algorithm on time expanded network.}
    \label{tab:table1}
    \scalebox{0.8}{
    \begin{tabular}{llll} 
    \hline
      \multicolumn{4}{c}{\textbf{Game Parameters}} \\ \hline
      \multicolumn{4}{c}{\textbf{Network Size: 4 Nodes, Crime Node: $0\_1$, Police Stations: $0\_2$, $0\_3$, $T_{max}$: 5, Exit Point: 4}}\\ \hline
      \parbox[t]{2cm}{Test Case} &  \parbox[t]{4cm}{Optimal Attacker strategy} & \parbox[t]{4cm}{Utility of the Final\\  Optimal Attacker Strategy} & \parbox[t]{2cm}{Run Time \\(Sec)}\\ \hline
       1 & [$0\_1, 3\_3, 5\_4$] & 0.0 & 0.004\\ \hline
       2 & [$0\_1, 3\_3, 5\_4$] & 0.0 & 0.004\\ \hline
       3 & [$0\_1, 3\_3, 5\_4$] & 0.0 & 0.0039\\ \hline
       4 & [$0\_1, 3\_3, 5\_4$] & 0.0 & 0.004\\ \hline
       5 & [$0\_1, 3\_3, 5\_4$] & 0.0 & 0.004\\ \hline
    \end{tabular}}
  \end{center}
\end{table}

To illustrate the proposed defender oracle, we next consider the six-node transportation network shown in Fig.~\ref{image-2}. In this example, node~6 represents the police station, node~1 is the crime node, node~5 is the exit node, and the maximum response time is $t_{\max}=6$. Here, $0\_6$ denotes node~6 at time $t=0$ in the multilayer network. Three attacker strategies, each assigned an equal probability, are used as input. Each attacker strategy is represented by a unique color, and every node--time pair visited by that strategy inherits the corresponding color in the multilayer network. The proposed approximation defender oracle is then applied to the time-expanded network (Fig.~\ref{image-3}) to generate a high-quality patrol strategy. The resulting defender patrols, illustrated by the green trajectories in Fig.~\ref{image-3}, follow the paths
\[
0\_6 \rightarrow 2\_3 \rightarrow 4\_4 \rightarrow 5\_5 \rightarrow 6\_5
\]
and
\[
0\_6 \rightarrow 2\_3 \rightarrow 4\_5 \rightarrow 5\_5 \rightarrow 6\_5.
\]
The corresponding defender strategies and computational times are reported in Table~\ref{tab:table2}. This illustrative example shows that the proposed approximation defender oracle efficiently generates defender patrol strategies that collectively cover all attacker strategies while satisfying the prescribed temporal constraints.

\begin{figure}[H]
\centering
\includegraphics[scale=0.5]{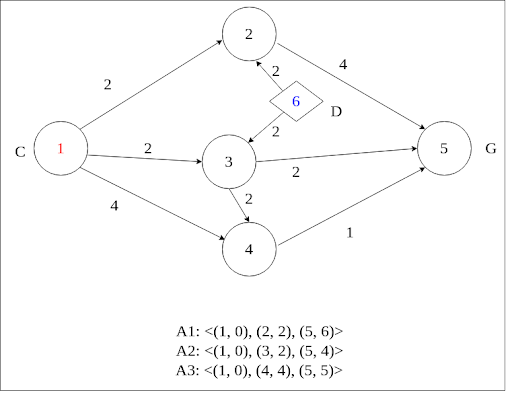}
\caption{Sample network for designing near-optimal defender strategy.}
\label{image-2}
\end{figure} 

\begin{table}[H]
  \begin{center}
   \caption{\\Defender strategy design using approximation defender oracle.}
    \label{tab:table2}
    \scalebox{0.8}{
    \begin{tabular}{llll} 
    \hline
      \multicolumn{4}{c}{\textbf{Game Parameters}} \\ \hline
      \multicolumn{4}{c}{\textbf{Network Size: 6 Nodes, Crime Node: $0\_1$, Police Station: $0\_6$, $T_{max}$: 6, Exit Point: 5}}\\ \hline
      \parbox[t]{2cm}{Test Case} &  \parbox[t]{4cm}{Final Defender strategy} & \parbox[t]{4cm}{Utility of the Final\\ Defender Strategy} & \parbox[t]{2cm}{Run Time \\(Sec)}\\ \hline
       1 & [$0\_6, 2\_3, 4\_4, 5\_5, 6\_5$] & 0.0 & 0.0086\\ \hline
       2 & [$0\_6, 2\_3, 4\_5, 5\_5, 6\_5$] & 0.0 & 0.0090\\ \hline
       3 & [$0\_6, 2\_3, 4\_5, 5\_5, 6\_5$] & 0.0 & 0.0081\\ \hline
       4 & [$0\_6, 2\_3, 4\_4, 5\_5, 6\_5$] & 0.0 & 0.0083\\ \hline
       5 & [$0\_6, 2\_3, 4\_5, 5\_5, 6\_5$] & 0.0 & 0.0081\\ \hline
    \end{tabular}}
  \end{center}
\end{table}

\begin{figure}[H]
\centering
\includegraphics[scale=0.36]{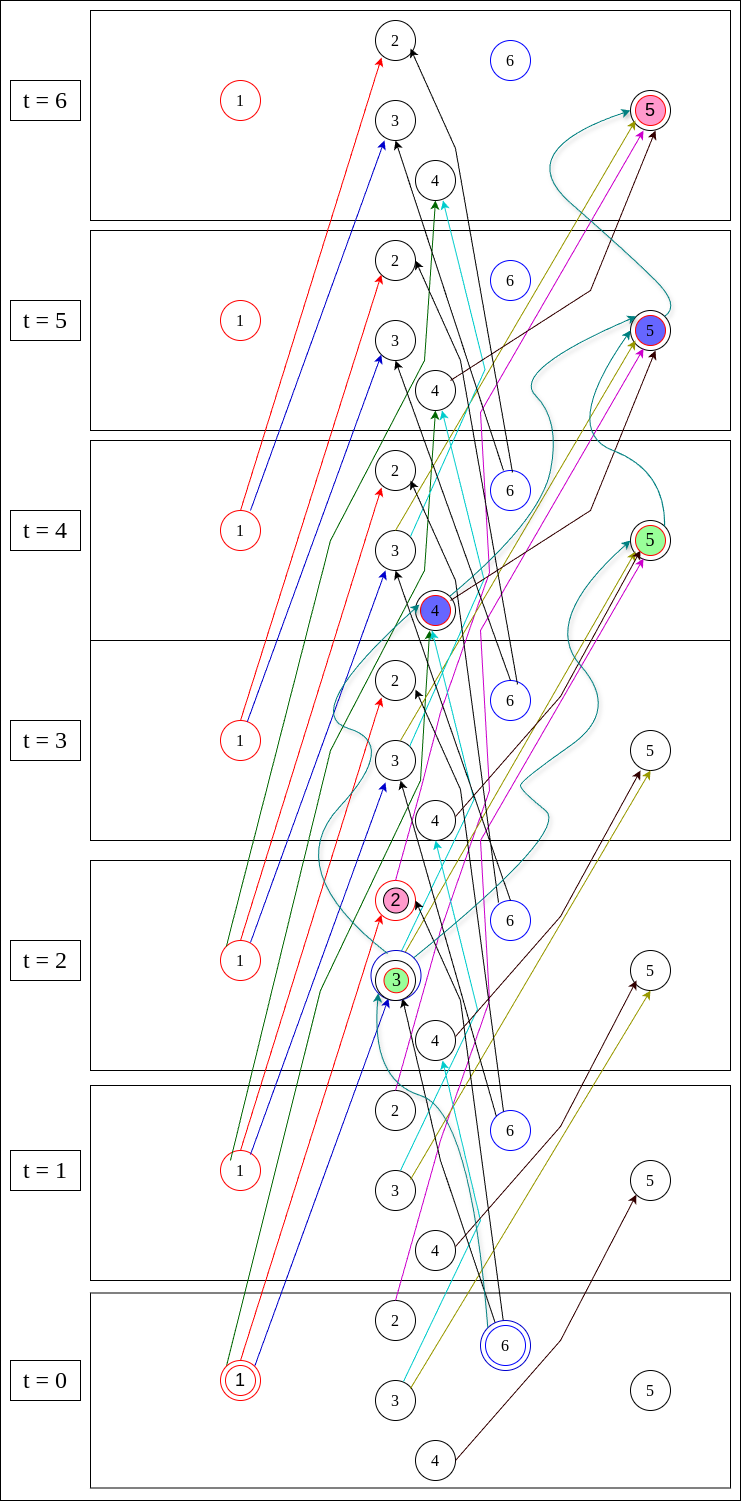}
\caption{Design of a multi-layer network for defender considering a sample network (Fig. \ref{image-2}).}
\label{image-3}
\end{figure} 

After illustrating the attacker and defender oracles on small synthetic examples, we evaluate the complete MLN-EIGS framework on a large real-world transportation network and compare its performance with the benchmark MILP-EIGS formulation. The benchmark framework employs the exact attacker and defender oracles, \textit{bestAo} and \textit{bestDo}, proposed by \citet{zhang2017optimal}. In both approaches, the Stackelberg game is considered to have converged when neither player can generate an improving strategy outside the current restricted strategy sets.

The computational experiments were conducted on the Central Kolkata transportation network, which consists of 461 nodes and 1,020 directed edges. The network was extracted from OpenStreetMap (OSM) and imported into the SUMO traffic simulator, as illustrated in Fig.~\ref{image-5}. The proposed MLN-EIGS framework was benchmarked against the exact MILP-EIGS formulation over 25 test instances generated by varying the initial defender locations, crime node, defender speed, attacker speed, and maximum response time. For each instance, we recorded the defender equilibrium utility, computational runtime, and the utility gap between the two approaches to evaluate both solution quality and computational efficiency.

\begin{figure}[H]
\centering
\includegraphics[scale=0.55]{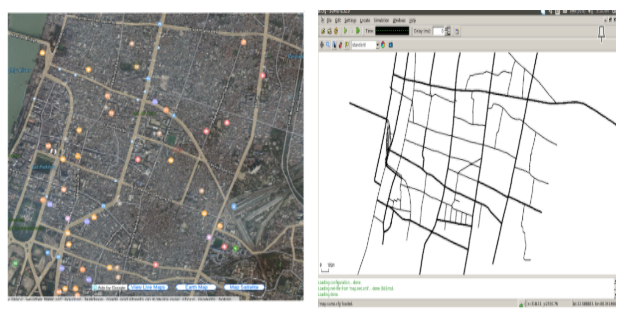}
\caption{Central Kolkata map (OSM) imported into SUMO for dataset generation}
\label{image-5}
\end{figure}

\begin{table*}[t]
\centering
\caption{Performance comparison between the benchmark MILP-EIGS and the proposed MLN-EIGS.}
\label{tab:comparison_results}
\scalebox{0.6}{
\begin{tabular}{cccccccccccc}
\hline
\textbf{Test} &
\textbf{Defender Locations} &
\textbf{Crime} &
\textbf{Defender} &
\textbf{Attacker} &
\textbf{Time} &
\textbf{Exit} &
\textbf{MILP} &
\textbf{MLN} &
\textbf{Utility} &
\textbf{MILP} &
\textbf{MLN} \\
\textbf{Case} &
\textbf{$(D_1,D_2,D_3,D_4)$} &
\textbf{Node} &
\textbf{Speed} &
\textbf{Speed} &
\textbf{Limit} &
\textbf{Points} &
\textbf{Utility} &
\textbf{Utility} &
\textbf{Gap} &
\textbf{Time (s)} &
\textbf{Time (s)}\\
\hline

1  & 132,277,365,409 & 127 & 1.0 & 4.0 & 200 & 13,73 & Infeasible & 0.742 & NA    & 0.587 & 0.292\\
2  & 14,66,107,446   & 21  & 1.2 & 1.0 & 90  & 13,73 & 1.000      & 0.615 & 0.385 & 7.250 & 0.653\\
3  & 6,17,100,46        & 107 & 5.0 & 5.0 & 40  & 13,73 & 0.000      & 0.000 & 0.000 & 2.157 & 0.400\\
4  & 14,66,107,446   & 21  & 1.2 & 1.0 & 90  & 13,73 & 1.000      & 1.000 & 0.000 & 7.250 & 0.653\\
5  & 7,55,109,448    & 15  & 1.2 & 1.5 & 140 & 13,73 & 0.000      & 0.000 & 0.000 & 5.331 & 0.497\\
6  & 81,111,118,169  & 55  & 1.8 & 1.6 & 180 & 13,73 & 1.000      & 0.796 & 0.204 &15.112 & 2.500\\
7  & 6,17,100,271    & 12  & 0.5 & 1.5 & 90  & 13,73 & 0.000      & 0.000 & 0.000 & 3.465 & 0.310\\
8  & 56,77,99,102    &103  & 2.0 & 2.0 & 220 & 13,73 & Infeasible & 0.864 & NA    & 0.579 & 0.281\\
9  & 1,53,59,76      & 53  & 1.5 & 2.2 & 220 & 13,73 & Infeasible & 0.719 & NA    & 0.568 & 0.278\\
10 & 2,57,70,72      & 58  & 2.0 & 2.5 & 250 & 13,73 & Infeasible & 0.582 & NA    & 0.552 & 0.283\\
11 & 3,25,35,50      & 25  & 2.5 & 3.0 & 230 & 13,73 & Infeasible & 0.907 & NA    & 0.560 & 0.286\\
12 & 49,51,60,89     & 60  & 2.2 & 3.5 & 240 & 13,73 & Infeasible & 0.674 & NA    & 0.568 & 0.282\\
13 & 26,32,38,40     & 32  & 3.0 & 4.0 & 240 & 13,73 & Infeasible & 0.488 & NA    & 0.578 & 0.286\\
14 & 37,44,45,46     & 45  & 3.5 & 4.5 & 240 & 13,73 & Infeasible & 0.821 & NA    & 0.563 & 0.280\\
15 & 61,62,63,64     & 61  & 4.0 & 5.0 & 230 & 13,73 & Infeasible & 0.396 & NA    & 0.552 & 0.284\\
16 & 78,95,195,205   & 95  & 1.0 & 5.0 & 160 & 13,73 & Infeasible & 0.731 & NA    & 0.573 & 0.288\\
17 & 87,96,228,379   &228  & 1.4 & 2.5 & 200 & 13,73 & 1.000      & 1.000 & 0.000 & 7.412 & 0.717\\
18 & 87,96,228,379   &228  & 1.4 & 2.5 & 200 & 13,73 & 1.000      & 1.000 & 0.000 & 7.412 & 0.717\\
19 & 10,52,56,77     & 10  & 1.1 & 2.4 & 180 & 13,73 & 1.000      & 1.000 & 0.000 & 7.284 & 0.904\\
20 & 265,361,430,448 &361  & 2.5 & 2.0 & 160 & 13,73 & 1.000      & 0.638 & 0.362 & 5.947 & 3.231\\
21 & 10,52,56,77     & 10  & 1.1 & 2.4 & 180 & 13,73 & 1.000      & 0.792 & 0.208 & 7.284 & 0.904\\
22 & 54,57,58,72     & 72  & 2.1 & 3.2 & 200 & 13,73 & Infeasible & 0.533 & NA    & 0.563 & 0.283\\
23 & 4,11,30,41      & 41  & 3.0 & 3.8 & 190 & 13,73 & Infeasible & 0.876 & NA    & 0.576 & 0.284\\
24 & 27,39,47,61     & 27  & 3.2 & 4.2 & 240 & 13,73 & Infeasible & 0.621 & NA    & 0.571 & 0.283\\
25 & 44,46,61,89     & 89  & 4.0 & 5.5 & 220 & 13,73 & 1.000      & 0.914 & 0.086 & 8.312 & 1.984\\

\hline
\end{tabular}}
\end{table*}

Table~\ref{tab:comparison_results} reports the computational performance of the proposed MLN-EIGS framework and the benchmark MILP-EIGS formulation over 25 test instances generated by varying the initial defender locations, crime node, defender speed, attacker speed, and maximum response time. For each instance, the defender equilibrium utility and computational runtime were recorded to evaluate both solution quality and computational efficiency.

For all feasible instances, MLN-EIGS produces defender utilities that closely approximate those obtained by the exact MILP-EIGS formulation. As shown in Table~\ref{tab:comparison_results} and Fig.~\ref{image-10}, identical defender utilities are obtained for Cases 3, 4, 5, 7, 17, 18, and 19, resulting in zero utility gap. For the remaining feasible instances, MLN-EIGS yields fractional defender utilities that remain close to the benchmark solutions, with the largest observed utility gap being 0.385 (Case 2). These fractional utilities arise naturally from the probabilistic multilayer network formulation, where the defender's utility is computed as the cumulative probability of interdiction along the attacker's path. In contrast, the MILP-EIGS formulation produces deterministic binary utilities corresponding to complete interception or successful escape.

Several test instances are reported as \emph{Infeasible} by the MILP-EIGS formulation. In these instances, no feasible attacker path exists from the specified crime node to any exit node within the prescribed response time. Consequently, the algorithm terminates during the preliminary feasibility verification stage without constructing or solving the complete mixed-integer optimization model. The corresponding runtimes (approximately 0.55--0.60 seconds) therefore represent only the computational effort required for input parsing, graph construction, and shortest-path feasibility verification rather than the solution of the MILP. In contrast, feasible instances proceed to iterative attacker and defender strategy generation, resulting in substantially higher computational times.

The runtime comparison shown in Fig.~\ref{image-9} further demonstrates the computational efficiency of the proposed MLN-EIGS framework. For every feasible instance, MLN-EIGS consistently requires significantly less computational time than the benchmark MILP-EIGS formulation. The computational advantage becomes increasingly pronounced for more challenging instances because the benchmark repeatedly solves mixed-integer optimization problems for both attacker and defender strategy generation within the double-oracle framework. In contrast, MLN-EIGS replaces the attacker oracle with a polynomial-time shortest-path computation on the multilayer network and approximates the defender oracle using an efficient greedy algorithm. Consequently, the proposed framework substantially reduces the overall computational effort while maintaining competitive equilibrium quality. Across all feasible instances, MLN-EIGS consistently reduced computational time by several factors while maintaining comparable defender equilibrium utilities.

Overall, the experimental results demonstrate that MLN-EIGS provides an effective and computationally efficient approximation to the exact MILP-EIGS formulation. Although the benchmark computes exact Stackelberg equilibrium solutions for feasible instances, MLN-EIGS achieves comparable defender equilibrium utilities with substantially lower computational cost. The results indicate that the proposed framework scales substantially better than the benchmark MILP formulation while preserving high solution quality, making it well suited for solving large-scale dynamic escape interdiction problems on real transportation networks.

\begin{figure}[H]
\centering
\includegraphics[scale=0.45]{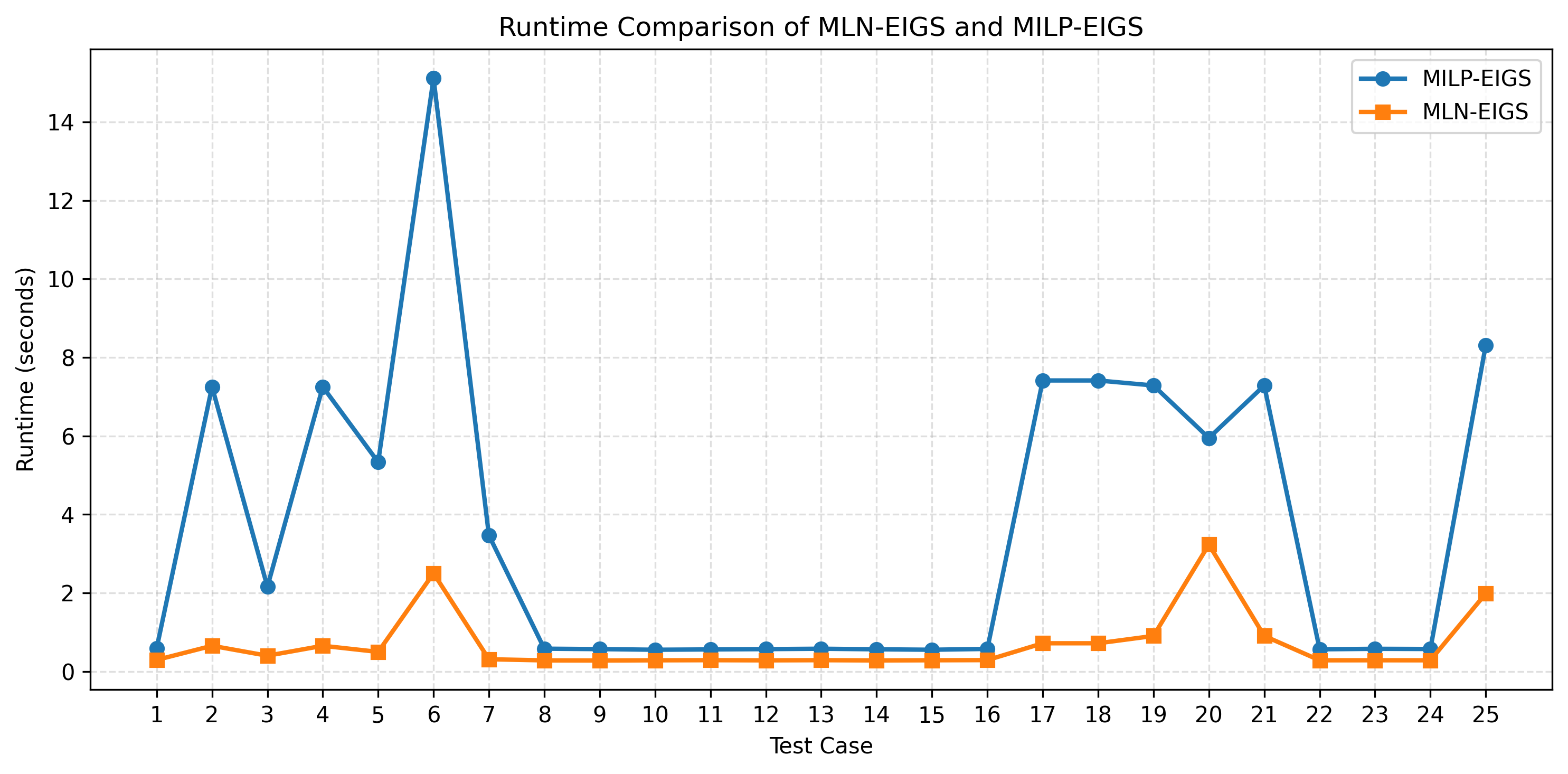}
\caption{Runtime comparison between the proposed MLN-EIGS and benchmark MILP-EIGS algorithms over 25 test instances.}
\label{image-9}
\end{figure}

\begin{figure}[H]
\centering
\includegraphics[scale=0.45]{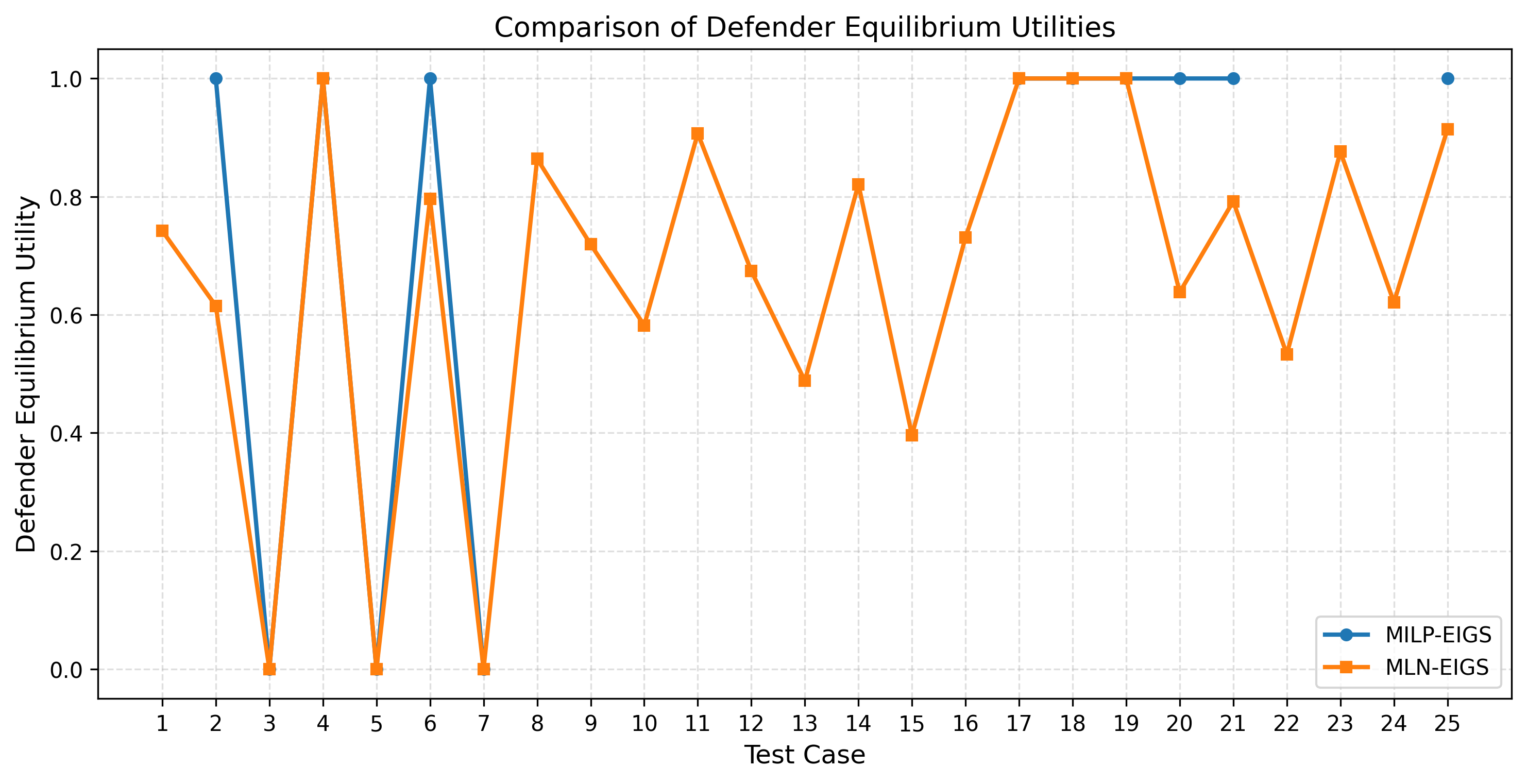}
\caption{Comparison of defender equilibrium utilities obtained by MLN-EIGS and MILP-EIGS over 25 benchmark instances.}
\label{image-10}
\end{figure}

\section{Conclusion}
\label{S:7}

This paper presented MLN-EIGS, a multilayer network-based Stackelberg framework for solving dynamic escape interdiction problems on time-dependent transportation networks. The proposed framework integrates temporal network expansion, probabilistic interdiction modeling, and Stackelberg game theory into a unified optimization framework. By representing attacker movements on a multilayer time-expanded network and reformulating the multiplicative escape-probability objective through a logarithmic transformation, the attacker best-response problem is reduced to an equivalent shortest-path problem that can be solved efficiently using Dijkstra's algorithm. Furthermore, the defender best-response problem is addressed through a polynomial-time approximation oracle within a double-oracle framework, enabling scalable computation of approximate Stackelberg equilibria on large transportation networks.

Computational experiments conducted on the Central Kolkata transportation network over 25 benchmark instances demonstrate that the proposed MLN-EIGS framework achieves defender equilibrium utilities that closely approximate those obtained by the benchmark MILP-EIGS formulation while consistently requiring substantially lower computational time. These results show that the multilayer network representation, together with the exact attacker oracle and approximation defender oracle, significantly improves computational scalability without sacrificing the quality of the computed equilibrium solutions.

The proposed framework provides a flexible foundation for several future research directions. Possible extensions include incorporating stochastic travel times, real-time traffic information, uncertain transportation environments, heterogeneous defender resources, and adaptive or learning-based attacker behavior. Another promising direction is the development of exact algorithms or approximation defender oracles with theoretical performance guarantees for the defender best-response problem, thereby strengthening both the computational and theoretical aspects of the proposed framework.

Overall, the proposed MLN-EIGS framework demonstrates that probabilistic multilayer network modeling provides a scalable and computationally efficient approach for computing approximate Stackelberg equilibria in large-scale dynamic escape interdiction problems, while achieving solution quality comparable to exact MILP-based methods.

\section*{Acknowledgments}

The authors are grateful to the members of the Multi-Agent Laboratory at Kyushu University for their valuable discussions and constructive comments. This research was supported by the Japan Society for the Promotion of Science (JSPS) through Grants-in-Aid for Scientific Research (KAKENHI).

\bibliographystyle{elsarticle-harv}
\bibliography{mlnEIGS}
\end{document}